\input{glyphtounicode}
\ifdefined\pdfinterwordspaceon
  \pdfinterwordspaceon
\fi
\documentclass[11pt]{article}
\usepackage[T1]{fontenc}
\usepackage{lmodern}
\usepackage[margin=1in]{geometry}
\usepackage{amsmath,amssymb,amsthm}
\usepackage{booktabs}
\usepackage[hidelinks]{hyperref}
\hypersetup{
  pdftitle={Sparse Disapproval Guarantees a Nonempty Hare Core},
  pdfauthor={Jiarui Fang},
  pdfsubject={Hare-core existence under at most two disapprovals per voter},
  pdfkeywords={approval voting, committee core, Hare quota, sparse disapproval, co-rank-two},
  pdfinfo={
    Affiliation={Boston University},
    ContactEmail={baymin@bu.edu},
    ORCID={0009-0006-9100-0445}
  }
}
\newtheorem{theorem}{Theorem}
\newtheorem{lemma}[theorem]{Lemma}

\newcommand{\N}{\mathcal N}

\title{Sparse Disapproval Guarantees a Nonempty Hare Core}
\author{Jiarui Fang\\
\small Boston University\\
\small \href{mailto:baymin@bu.edu}{\texttt{baymin@bu.edu}}\\
\small \href{https://orcid.org/0009-0006-9100-0445}{ORCID: 0009-0006-9100-0445}}
\date{September 2026}

\begin{document}
\maketitle

\begin{abstract}
An approval committee is Hare-core stable if no coalition meeting the Hare
quota can strictly improve by moving to another candidate set.  Whether every
approval election has such a committee remains open.  We prove nonemptiness
when each voter disapproves at most two candidates, with no bounds on the
numbers of candidates, seats, or voter types.  The result also permits
arbitrary positive rational voter weights.  Our deterministic rule represents
a committee by its missing set.  It first maximizes weighted coverage of
two-candidate disapproval sets and then maximizes total disapproval incidence.
An exact coverage inequality excludes targets one seat below the committee.
The incidence objective excludes unanimous equal-size targets, while targets
of size at most $k-2$ cannot improve any voter.  Two implementation-level
independent verifiers audit overlapping finite grids.  The symbolic proof,
not this bounded
enumeration, establishes the theorem's unbounded quantifiers.  The argument
identifies complement-side coverage as a tractable mechanism for a broad
parameter range within a sharply defined preference domain.
\end{abstract}

\section{Introduction}

An approval election consists of a finite voter set $\N$, a candidate set
$C$ of size $m$, approval sets $A_i\subseteq C$, and a target committee size
$k$.  A nonempty set $T\subseteq C$ blocks a size-$k$ committee $W$ at the
Hare quota when
\[
  \sum_{i:\,|A_i\cap T|>|A_i\cap W|} w_i\ge \frac{|T|}{k},
  \tag{1}
\]
where the positive rational voter weights sum to one.  A committee is in the
Hare core if it has no blocker.  Targets may overlap the committee.  For
example, when $k=3$, a singleton target reaches its Hare quota at one third
of the total voter weight.  If all those voters strictly gain, it blocks;
equality at the quota already permits blocking.

The core was introduced as a strong proportional-representation requirement
for approval committees by Aziz et al.~\cite{aziz2017}.  General deterministic
nonemptiness is open.  The publicly available results verified for this comparison
establish nonemptiness for $k\le8$ or $m\le15$~\cite{peters2025}, and for at
most seven weighted voter types~\cite{becker2026}.  Other positive results
impose interval or related preference structure~\cite{pierczynski2022}.
Approval-based apportionment guarantees core stability on a multi-copy party
domain~\cite{brill2024}.  That model constrains how candidates are grouped,
rather than how many candidates each voter disapproves.  Berker et
al.~\cite{berker2026} give an exact mixed-integer and duality framework with
additional special cases.

Our restriction is different from these parameter and structural bounds.
Write
\[
  H_i=C\setminus A_i
\]
for voter $i$'s disapproval set, and assume only $|H_i|\le2$.  This is a
dense-approval or near-consensus regime, but our motivation is structural
rather than empirical.  Each disapproval set is an empty set, a vertex, or
an edge.  Selecting a missing candidate set therefore becomes a weighted
vertex-edge coverage problem.  This complement-side representation permits
arbitrary $m$, $k$, and numbers of voter types.  A targeted primary-source
audit used ``bounded disapproval,'' ``two-veto,'' ``dense approvals,'' and
``complement rank'' terminology.  It found no equivalent theorem, so we
describe the result only as apparently new after a targeted audit.

\section{Selection rule and main theorem}

The endpoint $m=k$ is immediate by taking $W=C$.  Assume henceforth that
$m>k\ge2$ and put $s=m-k\ge1$.  A committee is represented by its missing
set $B=C\setminus W$, where $|B|=s$.  Define
\[
 e(B)=\sum_{i:\,|H_i|=2} w_i\,\mathbf 1[H_i\cap B\ne\varnothing],
 \qquad
 j(B)=\sum_i w_i|H_i\cap B|.
 \tag{2}
\]
Fix a total order on $C$.  Choose a size-$s$ set $B$ lexicographically
maximizing $(e(B),j(B))$.  If several sets remain, choose the first under the
induced lexicographic order on size-$s$ subsets, and return $W=C\setminus B$.

\begin{theorem}[Co-rank-two core theorem]\label{thm:main}
For every $m\ge k\ge2$, every approval election in which $|C\setminus
A_i|\le2$ for every positive-weight voter has a nonempty Hare core.  When
$m>k$, every committee associated with an $(e,j)$-maximizing missing set is
in the core; in particular, the fixed tie-break returns a core committee.
\end{theorem}

The rule is single-valued, explicit, and finite once the candidate order is
fixed.  We do not claim a polynomial-time
implementation; exhaustive evaluation over all $\binom ms$ missing sets
establishes the mathematical selection rule and is the method used by the
supplied verifier.

\section{Proof}

A target of size $t>k$ cannot block because its quota $t/k$ exceeds the
total voter weight one.  We may therefore restrict attention to a nonempty
target $T$ of size $t\le k$.  Write $T=C\setminus D$ and
$d=k-t$.  Then $|D|=s+d$.  Direct cancellation gives, for every voter,
\[
 |A_i\cap T|-|A_i\cap W|
 =-d+|H_i\cap D|-|H_i\cap B|.
 \tag{3}
\]

\subsection{Targets of size at most \texorpdfstring{$k-2$}{k-2}}

If $d\ge2$, the right side of (3) is at most $-d+|H_i|\le0$.  No voter
strictly gains, so such a target cannot block.

\subsection{Targets of size \texorpdfstring{$k-1$}{k-1}}

Now $d=1$.  Strict gain in (3) requires
\[
 |H_i\cap D|-|H_i\cap B|\ge2.
\]
Thus every gainer has $|H_i|=2$, has $H_i\subseteq D$, and has
$H_i\cap B=\varnothing$.  In particular, every gainer is a two-disapproval
voter not counted by $e(B)$.

\begin{lemma}\label{lem:prob}
A uniformly random size-$s$ set meets any fixed two-set with probability
\[
 p_{m,s}=1-\frac{k(k-1)}{m(m-1)}>\frac1k.
 \tag{4}
\]
\end{lemma}

\begin{proof}
The avoidance probability is
$\binom{m-2}{s}/\binom ms=k(k-1)/(m(m-1))$.  With $m=k+s$,
\[
 k[m(m-1)-k(k-1)]-m(m-1)
 =(k-1)((2s-1)k+s(s-1))>0,
\]
which is exactly the strict inequality in (4).
\end{proof}

Let $w_2$ be the total weight of two-disapproval voters.  Averaging (2) over
random $B'$ and using maximality gives $e(B)\ge p_{m,s}w_2$.  Hence every
size-$(k-1)$ target has gainer weight at most
\[
 w_2-e(B)\le(1-p_{m,s})w_2<1-\frac1k=\frac{k-1}{k}.
\]
This is strictly below its Hare quota.

\subsection{Targets of size \texorpdfstring{$k$}{k}}

Here $d=0$, so blocking requires every positive-weight voter to gain.  Suppose
that happened for some other size-$s$ missing set $D$.  Equation (3) yields
\[
 |H_i\cap D|>|H_i\cap B| \qquad\text{for every }i.
 \tag{5}
\]
An empty $H_i$ makes (5) impossible.  For each two-disapproval voter, (5)
implies that every pair hit by $B$ is also hit by $D$, and every pair missed
by $B$ is hit by $D$.  Therefore $e(D)\ge e(B)$.  Strict inequality
contradicts the first lexicographic objective.  If equality holds, multiply
(5) by the positive weight $w_i$ and sum over all voters.  This gives
$j(D)>j(B)$, contradicting the tie-break.  Thus no size-$k$ target blocks.
Together with the two preceding cases this proves Theorem~\ref{thm:main}.

\section{Exact verification}

The primary set-based verifier represents each anonymous profile as a tuple
of immutable disapproval sets.  For $m=4,5,6$, it tests every $2\le k<m$ and
every profile with total integer multiplicity at most four.  For $m=7$, the
multiplicity bound is three.  This grid contains 14 $(m,k)$ rows and checks
101,851 profile--parameter cases, 236,801 tied maximizers, and 12,020,946
literal committee--target inequalities.

The implementation-level independent verifier uses bitmasks, weighted
type-count vectors, and a separate enumeration architecture.  It repeats the
same 14 rows.  It also adds $(m,k)=(3,2)$ with multiplicity at most four and
all $2\le k<8$ at $m=8$ with multiplicity at most two.  It checks 106,620
profile--parameter cases, 265,371 tied maximizers, and 15,499,790 literal
inequalities.  Both implementations construct every tied maximizer, not
merely one.

All verifier arithmetic is integral, and both scripts refuse to run when
Python assertions are disabled.  The grids enumerate only the stated small
integer multiplicities.  They do not exhaust arbitrary denominators of
rational weights.  Lemma~\ref{lem:prob} and the symbolic proof establish
those unbounded quantifiers; the overlapping grids provide regression and
falsification evidence only.

\section{Scope and open boundary}

The theorem permits arbitrary positive rational weights, repeated ballots,
unapproved candidates, and targets overlapping the committee.  Equality in
the Hare quota counts as blocking throughout.  The theorem does not cover a
voter disapproving three candidates, does not resolve co-rank-three, and does
not imply unrestricted core nonemptiness.  Its contribution is a specialist
restricted-domain theorem, not evidence that dense approvals are empirically
universal.  The complement-side coverage and tie-break mechanism explains
why co-rank-two is tractable.  Whether a higher-order objective can control
deficit-one targets without losing the equal-size tie-break remains open.
The present averaging proof alone does not supply such an extension.

\section*{Reproducibility and AI assistance}

The package README gives a one-command exact replay and PDF build.  GPT-5.6 Sol
was used during proof exploration, verifier development,
literature-query formulation, and manuscript preparation.  The author
reviewed the statements, proofs, citations, and executable artifacts and
takes responsibility for the content.

\section*{Code availability}

The verification code and reproducibility materials are available from the
\url{https://github.com/Baymax-ray/approval-core-corank-two}.

\small
\bibliographystyle{unsrt}
\bibliography{references}
\end{document}